\documentclass[twocolumn,english,superscriptaddress,floatfix]{revtex4-2}

\usepackage{bm}% bold math
\usepackage[utf8]{inputenc}
\usepackage{amssymb}
\usepackage{amsthm}
\usepackage{amsmath}
\usepackage{thmtools, thm-restate}
\usepackage{physics}

\usepackage[lined,algoruled,algosection,linesnumbered,noend]{algorithm2e}
\SetAlCapSkip{0.5em}
\SetEndCharOfAlgoLine{} % Remove semicolons at end of line
\SetKwInOut{Data}{Data}
\SetKwInOut{Input}{Input}
\SetKwInOut{Output}{Out}
\SetKwFor{For}{for}{\string:}{endfor}
\SetKwFor{While}{while}{\string:}{endw}
\SetKwFor{ForEach}{foreach}{\string:}{endfch}
\SetKwFor{ParFor}{parallel for}{\string:}{endpfor}
\SetKwProg{Fn}{function}{\string:}{endfunction}

\usepackage{tikz}

\usepackage{qcircuit}

\usepackage{xcolor}
\usepackage[colorlinks=true,citecolor=blue,linkcolor=purple]{hyperref}

\usepackage[caption=false]{subfig}

\usepackage{mathtools} % For declaring math operations

\newcommand{\Mod}[1]{\ (\mathrm{mod}\ #1)}
\DeclareMathOperator{\rtop}{rt}
\DeclareMathOperator{\qrtop}{qrt} 
\DeclareMathOperator{\rtlop}{rt_{LOCC}}
\DeclareMathOperator{\rttop}{rt_{tele}}
\DeclareMathOperator{\rtsop}{rt_{\swap{}}}
\DeclareMathOperator{\diam}{diam}
\DeclareMathOperator{\advtg}{adv}
\newcommand{\maxadvtg}{\advtg^*}
\DeclareMathOperator{\poly}{poly}
\newcommand{\lapl}{\mathcal{L}}

\DeclarePairedDelimiterXPP\bigO[1]{O}{(}{)}{}{#1}
\DeclarePairedDelimiterXPP\bigomega[1]{\Omega}{(}{)}{}{#1}
\DeclarePairedDelimiterXPP\bigtheta[1]{\Theta}{(}{)}{}{#1}
\DeclarePairedDelimiterXPP\rt[1]{\rtop}{(}{)}{}{#1}
\DeclarePairedDelimiterXPP\qrt[1]{\qrtop}{(}{)}{}{#1}
\DeclarePairedDelimiterXPP\rtl[1]{\rtlop}{(}{)}{}{#1}
\DeclarePairedDelimiterXPP\rtt[1]{\rttop}{(}{)}{}{#1}
\DeclarePairedDelimiterXPP\rts[1]{\rtsop}{(}{)}{}{#1}

\newcommand{\divides}{\mid}

\newcommand{\swap}{swap}
\usepackage[capitalise,compress,nameinlink]{cleveref}
\crefname{section}{Sec.}{Secs.}
\crefrangelabelformat{equation}{\textup{(#3#1#4)}--\textup{(#5#2#6)}}

\usepackage{apptools}
\AtAppendix{\counterwithin{lemma}{section}}
\AtAppendix{\counterwithin{theorem}{section}}
\AtAppendix{\counterwithin{definition}{section}}
\newtheorem{theorem}{Theorem}[section]

\newtheorem{lemma}[theorem]{Lemma}

\newtheorem{proposition}[theorem]{Proposition}
\theoremstyle{definition}
\newtheorem{definition}[theorem]{Definition}

\begin{document}

\title{Quantum Routing with Teleportation}

\author{Dhruv Devulapalli}
    \email{ddhruv@umd.edu}
    \affiliation{Joint Center for Quantum Information and Computer Science, 
        NIST/University of Maryland, College Park, MD 20742, USA}
    \affiliation{Joint Quantum Institute, NIST/University of Maryland, 
        College Park, MD 20742, USA}

\author{Eddie Schoute}
    \affiliation{Computer, Computational, and Statistical Sciences Division, Los Alamos National Laboratory, Los Alamos, NM 87545, USA}
    \affiliation{Joint Center for Quantum Information and Computer Science, 
        NIST/University of Maryland, College Park, MD 20742, USA}
    \affiliation{Institute for Advanced Computer Studies, University of Maryland, College Park,
        MD 20742, USA}
    \affiliation{Department of Computer Science, University of Maryland, College Park,
        MD 20742, USA}

\author{Aniruddha Bapat}
    \affiliation{Joint Center for Quantum Information and Computer Science, 
        NIST/University of Maryland, College Park, MD 20742, USA}
    \affiliation{Joint Quantum Institute, NIST/University of Maryland, 
        College Park, MD 20742, USA}
    \affiliation{Lawrence Berkeley National Laboratory, Berkeley, CA 94720, USA}

\author{Andrew M.\ Childs}
    \affiliation{Joint Center for Quantum Information and Computer Science, 
        NIST/University of Maryland, College Park, MD 20742, USA}
    \affiliation{Institute for Advanced Computer Studies, University of Maryland, College Park,
        MD 20742, USA}
    \affiliation{Department of Computer Science, University of Maryland, College Park,
        MD 20742, USA}
        
\author{Alexey V.\ Gorshkov}
    \affiliation{Joint Center for Quantum Information and Computer Science, 
        NIST/University of Maryland, College Park, MD 20742, USA}
    \affiliation{Joint Quantum Institute, NIST/University of Maryland, 
        College Park, MD 20742, USA}

\date{\today}

\begin{abstract}
We study the problem of implementing arbitrary permutations
of qubits under interaction constraints
in quantum systems that allow for arbitrarily fast local operations and
classical communication (LOCC). In particular, we show examples of
speedups over \swap{}-based and more general unitary routing methods by distributing entanglement and using LOCC to perform
quantum teleportation. 
We further describe an example of an interaction graph for which teleportation gives
a logarithmic speedup in the worst-case routing time over \swap{}-based routing.
We also study limits on the speedup afforded by quantum teleportation---showing an $\bigO{\sqrt{N \log N}}$ upper bound
on the separation in routing time for any interaction graph---and give tighter bounds 
for some common classes of graphs.
\end{abstract}

\maketitle

\section{Introduction}
Common theoretical models of quantum computation assume that 2-qubit gates
can be performed between arbitrary pairs of qubits. 
However, in practice, scalable quantum architectures
have qubit connectivity constraints~\cite{arute, Monroe_Kim_2013}, 
which forbid long-range gates. 
These connectivity constraints are typically represented by a simple graph,
where vertices correspond to qubits, and edges indicate pairs of qubits
that can undergo 2-qubit gates. A quantum architecture with $N$ qubits is thus represented by a graph $G$ with $N$ vertices.
Circuits that use all-to-all connectivity must be transformed to new circuits that respect the architecture constraints specified by this graph.
Simple transformations introduce polynomial overhead in the worst case,
so it is crucial to lower this overhead. 

A natural approach to mapping circuits to respect interaction constraints is by permuting qubits
using \emph{routing} protocols.
Routing refers to the task of permuting packets of information, or
\emph{tokens}, on vertices of a graph. 
In \emph{quantum} routing, tokens are
data qubits, to be permuted on the graph specified by
the architecture's connectivity constraints. 
Previous work has used \swap{} gates to perform routing~\cite{cowtan, csu}, and
routing protocols from a classical setting using \swap{} 
gates~\cite{Alon1994, Chung1996, Zhang} can be naturally applied to the problem of routing quantum data as well.

Faster routing protocols can be obtained by using a wider range of quantum operations.
For example, Hamiltonian evolution can obtain a constant-factor speedup over \swap{}-based routing~\cite{reversals}.
More details of the comparisons and advantages of quantum routing models to classical routing models can be found in reference~\cite{quantum_routing}.
However, these approaches rely on locality-restricted unitary evolution, 
so the routing time is limited by the propagation speed of quantum information~\cite{lieb, quantum_routing}.

In this paper, we additionally allow for fast local operations, measurement and feedback (LOCC).
Since this model allows for fast classical communication across long distances, it is not similarly constrained by the propagation speed of quantum information.
For example, without prior shared entanglement, quantum teleportation over arbitrary distances
can be performed in constant depth by using entanglement swapping~\cite{swapping} in a quantum repeater protocol~\cite{repeater}, as shown in \cref{fig:teleportation}. Entanglement can also be distributed using quantum network coding protocols~\cite{Beaudrap_Herbert_2020}.
The ability to perform teleportation in constant depth immediately gives routing speedups over \swap{}-based methods and even over previous unitary quantum routing methods, since teleportation can be used to quickly exchange distant pairs of qubits.

\begin{figure}
    \includegraphics[width=\columnwidth]{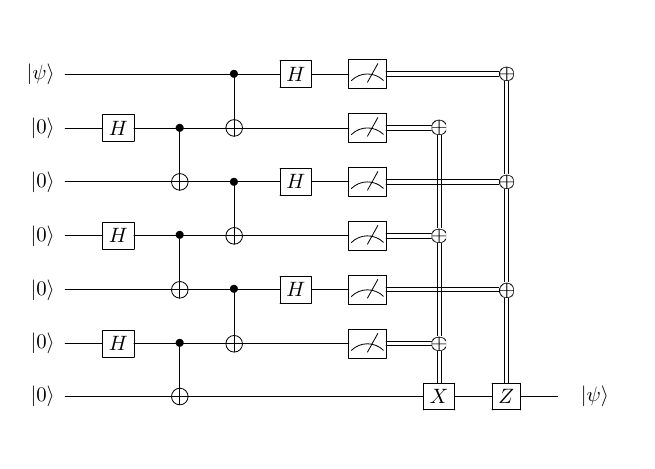}
    \caption{%
        Constant-depth long-range teleportation protocol on a path of 7 qubits. The $X$ and $Z$ gates are classically
        controlled by the parities of the two sets of measurement results. 
        This protocol can be extended to paths of any length without increasing the circuit depth.
    }
    \label{fig:teleportation}
\end{figure}

We show that using measurement and feedback to help prepare long-range entanglement can significantly decrease the time required for routing, even without using a large number of ancillas. In particular, we demonstrate the first superconstant speedup for quantum routing over swap-based routing in the setting where $\bigO{N}$ ancillas are allowed, showing a $\bigO{\log N}$ speedup for the hardest (i.e., worst-case) permutations.
Further, our main result proves the first non-trivial limits on the advantage of teleportation-based routing protocols
by an $\bigO{\sqrt{N \log N}}$ upper bound on the speedup over \swap{}-based routing.
Finally, we also show a new swap-based algorithm for sparse routing of $k$ qubits on a graph $G$ in time $\bigO{k+\diam(G)}$, where $\diam(G)$ is the diameter of $G$ 
(i.e, the maximum shortest-path distance between any pair of vertices).

LOCC has previously been useful to give low-depth implementations of specific unitaries, such as quantum fanout~\cite{Pham2013}, long-range operations on the surface code~\cite{Beverland_Kliuchnikov_Schoute_2022}, and the preparation of a wide range of entangled states~\cite{Piroli,rainbow,spt,verresen}.
In fact, previous work showed routing speedups by using ancillas~\cite{herbert} and by employing LOCC~\cite{rosenbaum}. Using teleportation,
\textcite{rosenbaum} showed a protocol that implements any permutation
in constant depth. However, Rosenbaum's protocol uses $\bigO{N^2}$ qubits
to perform permutations on $\bigO{N}$ qubits, so that only a negligible fraction of the qubits are data qubits. Engineering qubits is difficult, so it is preferable to use as many of them as possible as data qubits to enable larger computations.
Therefore, in this work we consider a more modest $\bigO{1}$ ancillas per data qubit (i.e. there are $\bigO{N}$ ancillas in total).
The availability of $\bigO{1}$ ancillas per data qubit is natural in some quantum systems, such as in NV center qubits~\cite{nvcenter}, 
quantum dots~\cite{dots}, and trapped ions~\cite{ions}.
To our knowledge, this is the first work to study quantum routing with measurement and feedback in the restricted ancilla setting.
By studying routing in this regime, we make progress on an open question posed by \textcite{herbert}, asking to what extent ancillas can be used to accelerate routing. 

%----------------discuss state transfer here--------------------%
Routing is more powerful than state transfer and entanglement distribution~\cite{bose_st, christandl_st}. For example, routing qubits from locally prepared Bell states can be used to generate long-range entanglement. The upper bounds in our work therefore also apply to these tasks.

Our work may be of interest to experimental efforts in systems which allow for mid-circuit measurements. In particular, the non-locality enabled by measurement and feedback makes large distances between qubits (i.e, large diameter connectivity graphs) less of a challenge for algorithm implementations.
Additionally, knowledge of (teleportation) routing may inform choices of connectivity in systems with these features.
In particular, our upper bounds can be used to compare routing overheads on different architectures based on their spectral and isoperimetric properties.

Furthermore, teleportation routing can also provide large advantages for specific permutations, which makes it useful for efficient implementations of algorithms on near-term architectures.
This is also of relevance to fault-tolerant quantum computation, as a major obstacle to the implementation of promising quantum error-correcting codes, such as qLDPC codes \cite{breuckmann_ldpc, panteleev2021quantum}, is their need for long-range syndrome measurements \cite{hong2023longrangeenhanced, Delfosse_Beverland_Tremblay_2021}. This can be alleviated by using teleportation to route together distant qubits from each syndrome. Further, protocols to prepare code states and implement logical operations in locality-restricted architectures are constrained by Lieb-Robinson bounds. Recent work \cite{friedman2022locality} has shown how the use of measurements can accelerate such tasks.  Routing schemes enabled by the use of teleportation can also be considered on fault-tolerant architectures, such as, for example, to perform logical circuits across surface code patches. More generally, the use of measurement and feedback enables speedups from the ability to implement long-range interactions quickly, which can also make algorithms much easier to run on near term architectures~\cite{Pham2013}.

%-------------------------------------------------------------------------%
Our paper is organized as follows. 
After introducing the models in \cref{sec:prelim},
we discuss known upper and lower bounds 
on the routing time for both \swap{}-based and teleportation routing in \cref{known_routing_bounds}. We also
introduce an improved algorithm for sparse routing (i.e, routing of a small subset of tokens) with \swap{}s and ancillas.
In \cref{example}, we use teleportation to speed up specific permutations. In \cref{advantage}, we compare teleportation routing to
\swap{}-based routing for \emph{arbitrary} permutations, and we give an example of a $\bigO{\log N}$-factor speedup over \swap{}-based routing. 
In \cref{bounds}, we show an $\bigO{\sqrt{N \log N}}$ upper bound on the speedup of teleportation routing over \swap{}-based routing for all graphs, and show tighter bounds
for some common classes of graphs.
Finally, we conclude in \cref{discussion} with a discussion of the results and some open questions.

%-------------------------------------------------------------------------%
%-------------------------------------------------------------------------%
\section{Preliminaries}\label{sec:prelim}
We consider architectures consisting of $N$ data qubits connected according to a simple graph 
$G$ (with vertex set $V(G)$ and (undirected) edge set $E(G)$), 
where an edge $(u,v) \in E(G)$ represents a connection between qubits $u,v \in V(G)$, and $|V(G)|=N$. 
We consider only \emph{connected} graphs, i.e., graphs in which there is a path from any vertex
to any other vertex.

We assume there are a constant number of ancillary qubits per data qubit
that can interact only with the data qubit.
Further, we assume that disjoint
two-qubit gates can be performed between adjacent qubits in depth 1.
Up to a constant overhead, this is equivalent to having fast (instantaneous) ancilla interactions
since any unitary on the data qubit and ancillas can be decomposed into a constant number of two-qubit gates.
As our results are asymptotic, they are insensitive to a constant overhead.

Ancillary qubits corresponding to different data qubits are not directly connected. However, 
gates between ancillary qubits of neighboring vertices can be performed in depth 1 by swapping ancillas with their corresponding data qubits, 
performing the desired 2-qubit gate between data qubits,
and swapping again with the ancillas.
This model can be implemented in realistic quantum architectures with
attached ancillas~\cite{nvcenter, dots, ions} as well as architectures with grid connectivity such as superconducting qubits~\cite{IBM,arute}. For example, \cref{fig:grid} shows an architecture where ancillas are interspersed with data qubits on a grid. This can be represented in our model as \cref{fig:anc}. Both models are equivalent and can simulate each other with only constant depth overhead.

\begin{figure}
    \subfloat[]{\label{fig:grid} \includegraphics[width=0.45\columnwidth]{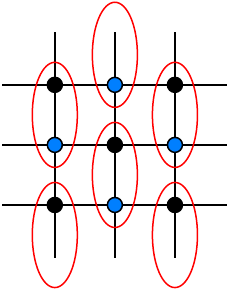}}
    \quad
    \subfloat[]{\label{fig:anc} \includegraphics[width=0.4\columnwidth]{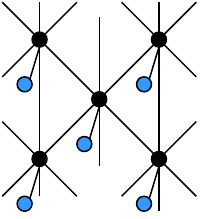}}
    
    \caption{(a) A grid architecture with ancillas (blue) interspersed between data qubits (black). The red ovals indicate which ancilla corresponds to each data qubit. (b) An equivalent architecture in our model.}
\end{figure}

The task of routing involves permuting data qubits on the graph.
We use the notation
\begin{equation}
    \{(1,\pi(1)), (2, \pi(2)), \dots, (N, \pi(N))\}
\end{equation} to denote a permutation
on $N$ vertices, where $\pi(i)$ is the vertex to which we must move the $i$th qubit.
We also write
\begin{equation}
    [N] \coloneqq \{1, 2, \dots N\}.
\end{equation}

We consider the following models of routing.
\begin{enumerate}
    \item Swap routing: In this model, the only allowed gates between adjacent qubits are \swap{} gates.
    \item LOCC routing: 
    In this model, we are allowed to perform arbitrary 2-qubit gates on disjoint pairs of qubits
    in a single time step. Further, in the same time step, we are allowed to perform single-qubit measurements (on data and ancilla qubits)
    and adaptively apply arbitrary single-qubit gates. 
    We refer to this as \emph{fast measurement and feedback.}
    Gates in later time steps can be applied adaptively, conditioned on all previous measurement
    results.
    
    \item Teleportation routing: 
     In this model, data qubits can be teleported along disjoint paths to ancilla registers at arbitrary distances in depth 1. Using this ability, a \swap{} between the ends of a path can be performed in constant depth. Note that the qubits along a teleportation path cannot be involved in any other operations during a round of teleportation.
    However, teleportation between multiple pairs of qubits can be performed in parallel 
    if there exist paths for each pair that have no more than a constant
    number of intersections per vertex, since we allow a constant number of ancilla qubits per data qubit.
     This model is a specialization of LOCC routing as the ability to 
    perform fast measurement and feedback allows us to perform quantum teleportation,
    transporting a single qubit to any vertex in constant depth.
    The entanglement required for quantum teleportation is produced using an entanglement swapping protocol~\cite{swapping}, as depicted
    in \cref{fig:teleportation}. A \swap{} between the ends of a path can be performed by teleporting the qubit at each end to the opposite end, or by performing gate teleportation~\cite{gottesman_chuang} of a \swap{} gate.

\end{enumerate}

We are particularly interested in the \emph{routing time} $\rt{G,\pi}$, which 
is the minimum circuit depth to perform the permutation $\pi$ on the data qubits of $G$.
The worst-case routing time of a graph $G$ is 
\begin{equation}
    \rt{G} \coloneqq \max_{\pi \in \mathcal{S}_{N}} \rt{G, \pi}
\end{equation}
where $\mathcal{S}_N$ is the symmetric group, i.e., the group of all permutations of $N$ elements.
We let $\rtt{G}$ denote the routing time in the teleportation model, 
$\rtl{G}$ denote the routing time in the LOCC model,
and $\rts{G}$ denote the routing time in the \swap{} model.

%-------------------------------------------------------------------------%
%-------------------------------------------------------------------------%
\section{Bounds on routing time}
\label{known_routing_bounds}
In this section, we discuss known bounds on the routing time for both \swap{} and LOCC routing.

\subsection{Lower bounds}
If a permutation can only be implemented by sending a large number of tokens through a small number 
of vertices, then any circuit for performing it must have high depth, since each vertex can only hold one token at a time. This gives a natural lower bound on the routing time.
To formalize this, we consider the \emph{vertex expansion} (or \emph{vertex isoperimetric number}) $c(G)$ of a graph $G$, defined as follows.

\begin{definition}
\label{def:vertex_expansion}
The vertex expansion of a graph $G$ is
\begin{equation}
c(G) \coloneqq \min \limits_{X \subseteq V(G)} \frac{|\delta X|}{\min\{|X|, |\overline{X}|\}},
\end{equation}
where 
\begin{equation}
    \overline{X} = V(G) - X
\end{equation} 
is the complement of $X$, and
\begin{equation}
    \delta X = \{v \in \overline{X}
\mid \exists \, u \in X \text{ s.t.\ } (u,v) \in E(G)\}
\end{equation}
is the vertex boundary of $X$.
\end{definition}
Note that $c(G) \leq 1$:
\begin{align}
&\min \limits_{X \subseteq V(G)} \frac{|\delta X|}{\min\{|X|, |\overline{X}|\}} \\
&= \min \limits_{X \subseteq V(G)}  \left( \frac{|\delta X|}{\min\{|X|, |\overline{X}|\}}, \frac{|\delta \overline{X}|}{\min\{|X|, |\overline{X}|\}} \right)
\\&=\min \limits_{X \subseteq V(G)} \frac{\min(|\delta X|, |\delta \overline{X}|)}{\min\{|X|, |\overline{X}|\}} \\ &\leq 1.
\end{align}
In addition, for a connected graph, since $|\delta X| \geq 1$ and  $\min\{|X|, |\overline{X}|\} \leq \frac{N}{2}$ for any $X$, we have $c(G) \geq \frac{2}{N}$
Therefore, for connected graphs, $c(G) \in \left[ \frac{2}{N}, 1 \right]$.

Any connected simple graph $G$ satisfies the following.

\begin{theorem}[Isoperimetric lower bound~\cite{quantum_routing}]
\label{cut_bound}
\begin{align}
    \rtl{G} &\geq \frac{2}{c(G)} -1.
\end{align}
\end{theorem}

Since  $\rts{G} \geq \rtt{G} \geq \rtl{G}$, this lower bound applies to \swap{}- and teleportation-based routing as well.

%-------------------------------------------------------------------------%

We can also lower bound the \swap{}-based routing time by the diameter of the graph 
(i.e, the maximum shortest-path distance between any pair of vertices) since swapping two vertices at distance $d$ requires a swap circuit of depth at least $d$.

\begin{theorem}[Diameter lower bound]
\label{diam_bound}
\begin{equation}
    \rts{G} \geq \diam(G).
\end{equation}
\end{theorem}

Note that this bound does \emph{not} apply to teleportation or LOCC routing.

%-------------------------------------------------------------------------%
\subsection{Upper bounds}
On any graph, a classical \swap{} algorithm can route on an $N$-vertex tree in depth $O(N)$ 
\cite{Zhang}. Recall that we only consider connected graphs, so we can always route on a spanning tree with \swap{}s in depth $O(N)$. 
We thus have the following upper bounds.

\begin{theorem}
\label{\swap{}n}
For any $N$-vertex connected graph $G$,
\begin{equation}
    \rts{G} = \bigO{N}
\end{equation}
\end{theorem}

This bound also implies that
$\rtt{G} = \bigO{N}$ and 
$\rtl{G} = \bigO{N}$.

We can prove a tighter bound for
\emph{sparse} routing.
Let $\rts{G, k}$ denote the worst-case routing time on $G$ over permutations that move at most $k$ tokens.
Using reversals,~\cite{quantum_routing} gives a routing algorithm that takes depth $\bigO{\diam(G) + k^2}$. 
We improve this result, using \swap{}s with ancillas, to show the following.

\begin{restatable}[Sparse routing]{theorem}{srt}\label{th:sparse}
For any $N$-vertex connected simple graph $G$ and $k \in [N]$,
\begin{equation}
    \rts{G, k} = \bigO{\diam(G) + k}.
\end{equation}
\end{restatable}

\begin{proof}[Proof sketch]
Call all tokens $v$ with $\pi(v) \neq v$ \emph{marked}. There are $k$ marked tokens.
There are three main steps in our algorithm:
\begin{enumerate}
    \item Hide all unmarked tokens in the ancillas by performing \swap{}s. Route the $k$ marked tokens to span a tree subgraph in time $\bigO{\diam(G)}$.
    \item Permute the $k$ tokens on the tree subgraph, using the procedure from~\cite{Zhang}, in time $\bigO{k}$.
    \item Reverse the first step, thereby moving the $k$ tokens from the subgraph to the appropriate target locations in time $\bigO{\diam(G)}$. Restore the unmarked tokens from the ancillas. 
\end{enumerate}
See \cref{sparse_tree} for the full proof.
\end{proof}

%-------------------------------------------------------------------------%
%-------------------------------------------------------------------------%
\section{Faster permutations with teleportation}
\label{example}
The ability to perform teleportation immediately suggests possibilities for speedups over 
\swap{}-based routing. Swap-based routing must obey the diameter lower bound (\cref{diam_bound}), so permutations that involve 
long-range \swap{}s (e.g., between diametrically separated pairs of vertices) should be 
sped up by teleportation.

We define the \emph{teleportation advantage} for a specific permutation to quantify
this speedup:
\begin{equation}
    \advtg(G, \pi) \coloneqq \frac{\rts{G,\pi}}{\rtt{G, \pi}}.
\end{equation}

We now consider the following permutation on the path graph $P_N$:
$\pi_{\mathrm{diam}} = \{(1,N), (2,2), \dots, (N-1, N-1), (N,1)\}$ (see \cref{fig:line_diam}).

\begin{figure}
    \subfloat[Diameter-length permutation $\pi_\mathrm{diam}$ (shown by the red double-sided arrow) on the path graph  $P_N$ (shown in black).]{\label{fig:line_diam} \includegraphics[width=0.9\columnwidth]{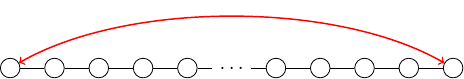}}
    \hfill %%
    \subfloat[Rainbow permutation $\pi_\mathrm{rainbow}^\alpha$ (shown by red double-sided arrows) on the path graph $P_N$ (shown in black).]{\label{fig:rainbow} \includegraphics[width=0.9\columnwidth]{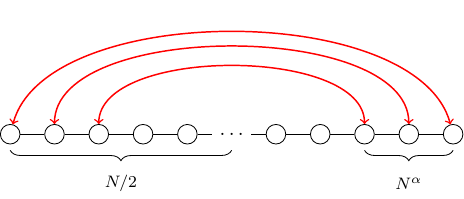}}
    
    \caption{Permutations on a 1D lattice}
\end{figure}

By the diameter lower bound, this permutation takes depth $\Omega(N)$ with \swap{}s. 
However, with teleportation it takes depth $\bigO{1}$, showing that
$\advtg(P_N, \pi_{\mathrm{diam}}) = \bigomega{N}$.

This further generalizes to permutations that require multiple long-range \swap{}s. 
For example, consider a \emph{rainbow} permutation $\pi_{\mathrm{rainbow}}^\alpha$, as depicted in \cref{fig:rainbow}.
This permutation involves performing $N^{\alpha}$ \swap{}s across a 1D lattice for some $\alpha\in [0,1]$. With \swap{}s, this takes
depth $\Theta(N)$ by the diameter bound, but with teleportation 
it takes depth $N^{\alpha}$, by a procedure that simply teleports
each pair into place sequentially. 
This gives a polynomial advantage: $\advtg(P_N, \pi_{\mathrm{rainbow}}^\alpha) = \bigO{N^{1-\alpha}}$.

These permutations allow speedups bounded by the diameter of the graph. Any single teleportation step can be simulated by \swap{}s in depth $\bigO{\diam(G)}$, by simply swapping along the shortest path between the initial qubit and the final destination. Intuitively, one might 
therefore expect that teleportation routing could achieve at most a diameter-factor speedup. 
However, there exist some graphs and permutations for which we can obtain even larger speedups. 
Teleportation speedups are not limited by the graph diameter since teleportation protocols
can utilize multiple longer paths together to avoid intersections.

\begin{figure}
    \includegraphics[width=0.6\columnwidth]{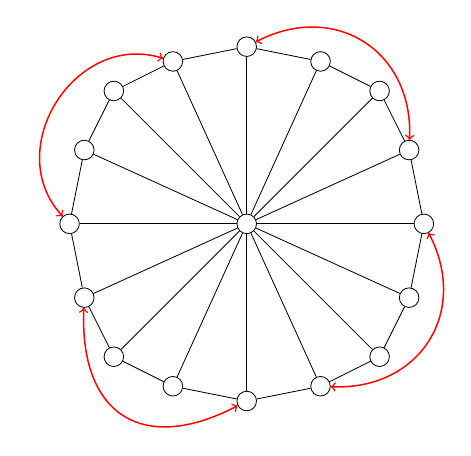}
    \caption{Permutation $\pi_{\mathrm{wheel}}^l$ (shown by red double-sided arrows) that exchanges $l$ pairs of vertices on the wheel graph $W_{N+1}$ (shown in black).}
    \label{fig:wheel}
\end{figure}

To illustrate this, consider the example of a 
\emph{wheel} graph $W_{N+1}$, as shown in \cref{fig:wheel}.
The $(N+1)$-vertex wheel graph, with central vertex $N+1$, has edges
\begin{equation}
    E(W_{N+1})=\{(u,v) \mid u-v = 1\Mod{N} \text{ or } v=N+1\}.
\end{equation}
The diameter of $W_{N+1}$ is 2.
On this graph, consider the permutation (shown in red in \cref{fig:wheel})
\begin{equation}
\begin{aligned}
    \pi_{\mathrm{wheel}}^l\coloneqq\{&(1, N/l), (N/l+1, 2N/l), \ldots, \\ 
    & (N-N/l+1, N)\}
\end{aligned}
\end{equation}
that exchanges $l$ pairs of vertices spaced along the ``rim'' of the wheel (assume $l \divides N$).
For \swap{}-based algorithms, this can be done in depth $\min\{3l, N/l-1\}$ by routing the qubits sequentially through the central vertex or routing them in parallel along the ``rim'', whichever is faster.

This is optimal up to constant factors, by the following reasoning. If there exists a data token that does not pass through the central node, the routing time must be at least $N/l-1$, which is the travel distance along the rim. On the other hand, if every data token passes through the central node, then there must be at least $2l$ steps in the algorithm. Therefore
\begin{equation}
    \rts{W_{N+1}, \pi_{\mathrm{wheel}}^l} \geq \min\{2l,N/l-1\}.
\end{equation}

However, in the teleportation routing model, this permutation can be performed in constant depth
by performing $l$ teleportations in parallel along non-intersecting paths on the wheel rim.
Therefore,
\begin{equation}
    \rtt{W_{N+1}, \pi_{\mathrm{wheel}}^l} = \bigO{1}.
\end{equation}
Setting $l = \sqrt{N/2}$, we obtain a maximum teleportation 
advantage $\advtg(W_{N+1}, \pi_{\mathrm{wheel}}^l) = \Theta(\sqrt{N})$ for this class of permutations, even though $\diam(W_{N+1}) = \bigO{1}$. Teleportation therefore enables super-diametric speedups.

%-------------------------------------------------------------------------%
%-------------------------------------------------------------------------%
\section{Teleportation Advantage}
\label{advantage}
While $\pi_{\mathrm{diam}}$, $\pi_{\mathrm{rainbow}}^\alpha$, and $\pi_{\mathrm{wheel}}^l $ allow for teleportation speedups, they are not
the worst-case permutations on their respective graphs. For example, consider the full reflection
on the line graph, i.e., a rainbow permutation with $\alpha = 1$. 
This permutation requires depth $\Theta(N)$ for 
both \swap{}- and teleportation-based routing. Similarly, on the wheel graph with an even number of vertices, the permutation $\pi$ with 
$\pi(i) = i+\lfloor N/2 \rfloor \bmod N$ for all $i \in [N]$ 
requires depth $\Theta(N)$
for both types of routing as well. 
Thus, although these graphs have teleportation speedups for specific permutations,
there is no separation between their \swap{} and teleportation routing numbers.

To compare the relative strength of the teleportation routing model to the \swap{}-based
routing model for all permutations, we aim to understand how much teleportation improves 
worst-case permutations. We measure the relative strength of the teleportation model by 
the separation in teleportation and \swap{}-based routing numbers,
which we define as the \emph{worst-case teleportation advantage}:
\begin{equation}
    \advtg(G) \coloneqq \frac{\rts{G}}{\rtt{G}}.
\end{equation}
Note that this is \emph{not} the worst-case ratio of routing numbers for a single specific
permutation, i.e., $\advtg(G)$ is not necessarily the same as $\max_{\pi} \advtg(G, \pi)$. (Indeed, as discussed above, these two quantities differ for the path and wheel graphs.) Instead,
$\advtg(G)$ can be thought of as the speedup teleportation provides for the general
task of routing on a particular graph in the worst case, rather than for implementing a specific permutation.
It also allows us to compare different graphs:
teleportation routing offers greater worst-case guaranteed
speedups on graphs with higher $\advtg(G)$.

\begin{figure}
    \includegraphics[width=\columnwidth]{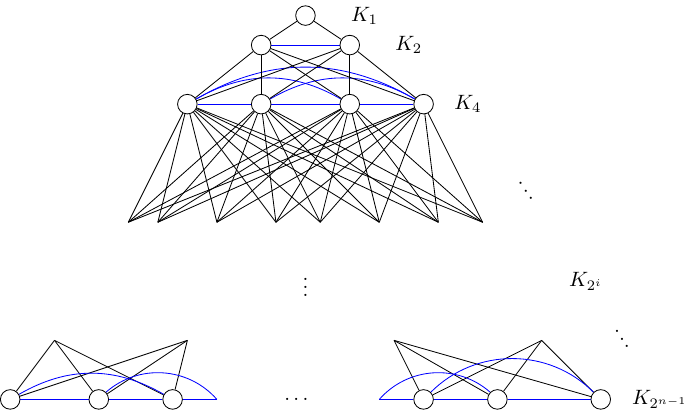}
    \caption{The graph $L(n)$. The black lines show edges between layers,
    while blue lines show edges within a layer (colored for visibility).}
    \label{fig:ladder}
\end{figure}

It is not immediately obvious that we should expect $\advtg$ to be greater than 1 for any graph.
However, we now describe a graph that
does offer a worst-case speedup for teleportation.
This graph, which we denote by $L(n)$ (with $N=2^n-1$ vertices), 
has $\advtg(L(n))=n=\log_2 (N+1)$.
The graph $L(n)$ (depicted in \cref{fig:ladder}) has
\begin{align}
    V(L(n)) &= \{(r, i) \mid r \in [n], i \in [2^{r-1}]\}
\end{align}
and 
\begin{align}
    E(L(n)) =
    \begin{aligned}[t]
    & \begin{aligned}[t]
        \{&((r, i_1), (r, i_2)) \mid \\
        &\quad r \in [n],\, i_1 < i_2 \in [2^{r-1}]\}
    \end{aligned}
    \\
    & \cup \begin{aligned}[t]
        \{&((r_1, i_1), (r_2, i_2)) \mid \\
        &\quad r_2-r_1 = 1,\, i_1,i_2 \in [2^{r-1}]\}.
    \end{aligned}
    \end{aligned}
\end{align}
In words, $L(n)$ is a ladder formed by arranging complete graphs $K_{2^k}$ for $k\in\{0,1,\ldots, n-1\}$ in horizontal layers, and then connecting every vertex in a given layer with every vertex one layer above or below. The total number of vertices in this graph is 
\begin{equation}
    N = \sum\limits_{k=0}^{n-1} 2^k = 2^n-1.
\end{equation}
The diameter of $L(n)$ 
is exactly $n-1=\log_2(N+1) -1$. \cref{diam_bound} then implies
\begin{equation}
    \rts{L(n)} = \Omega(\log(N)).
\end{equation} 

With teleportation, we show that routing can be performed in depth $\bigO{1}$. 
The key idea behind the teleportation protocol is that
every layer of $L(n)$ has one more node than all the layers above it together. 
This allows us to identify a unique node in each layer corresponding to
any node from a higher layer. We can then route tokens
by simply teleporting along the path formed by the unique nodes from each layer,
corresponding to the source vertex of the token to be routed.

This teleportation routing procedure establishes the following.

\begin{proposition}
$\rtt{L(n)} = O(1)$. 
\end{proposition}
\begin{proof}
For any permutation $\pi\in S_{2^n-1}$, we construct a set of paths  $\{P(u,\pi(u)) \mid u\in V(L(n))\}$ between every node and its destination such that each vertex of the graph belongs to at most four paths in the set.

Label every vertex in the graph with an $n$-bit address as follows. 
To every node in the subgraph 
$K_{2^i}$ (corresponding to layer $i+1$ of the ladder), assign a unique integer $u$ in the range $[2^i, 2^{i+1}-1]$. 
(Since the layer is a complete graph, the order within a layer is arbitrary.) 
Equivalently, we may refer to node $u$ by its binary representation $b(u)$, which is an $(i+1)$-bit string with a leading 1, i.e., of the form $b(u) = (1\ldots)$. 

For any vertex $u\in V$, define $r(u,i)$ to be the vertex whose address is $(10^{i-1}b(u))$,
i.e., the address of $u$ appended to a leading 1 $i$ places to the left. Note that $r(u,i+1)$ is adjacent to $r(u,i)$ and lies in the layer immediately below $r(u,i)$. Define $r(u,0)=u$.

Now, given two vertices $u,v$ separated by a distance $d$, define a \emph{canonical path} $P(u,v) = P(v,u)$ as the sequence of the following nodes: $(u, r(u,1),\ldots, r(u,d-1), v)$, where we assume $u < v$ without loss of generality. If $d=1$, then $P(u,v)=(u,v)$. We now show that for any permutation $\pi\in S_{2^n-1}$, the set of canonical paths $\{P(u,\pi(u)) \mid u\in V(L(n))\}$ intersects any vertex at most four times. 

Fix an arbitrary vertex $v$. By construction, $v$ lies in $P(v,\pi(v))$ and $P(v,\pi^{-1}(v))$. Now suppose a path $P(u,\pi(u))$ passes through $v \notin \{u,\pi(u)\}$. Then either $u < v< \pi(u)$ or $\pi(u) < v < u$. Without loss of generality,
we assume the former. Since $P(u,\pi(u))$ is canonical, $b(v) = (10^{i}b(u))$
for some $i\ge 0$. Suppose a different path $P(u', \pi(u'))$ also intersects $v$. 
Then there are two cases to consider: $u' < \pi(u')$ and $u' > \pi(u')$. In the first case, $b(v) = (10^{i'} b(u'))$. This is only possible when $i=i'$ and $u = u'$, which implies that $P(u',\pi(u')) = P(u,\pi(u))$ (giving one intersecting path at $v$). In the second case, the same reasoning implies that $u = \pi(u')$. In this case, there are two intersecting paths $P(u,\pi(u))$ and $P(\pi^{-1}(u), u)$ at $v$. Therefore, in addition to $P(v,\pi(v))$
and $P(v,\pi^{-1}(v))$, at most two other paths can intersect at $v$, giving a total of at most four paths.

Finally, construct one Bell pair for every edge in every canonical path $P(u,\pi(u))$,
using distinct local ancillas for every pair. The number of Bell pairs shared at any vertex
is at most $6 = O(1)$, requiring 6 local ancillas per vertex. 
Using the standard repeater protocol (\cref{fig:teleportation})
along each canonical path, one can then carry out simultaneous teleportation of 
all data qubits to their destination vertices $v\mapsto \pi(v)$ in constant depth. Therefore, any permutation of the qubits can be implemented
in depth $O(1)$. 
\end{proof}

%-------------------------------------------------------------------------%
%-------------------------------------------------------------------------%
\section{Bounding the Teleportation Advantage}
\label{bounds}
In the previous section, we described a graph with logarithmic teleportation advantage.
In this section, we examine limits on the teleportation advantage.
In order to understand the power of teleportation in general, we specifically 
aim to bound the \emph{maximum teleportation advantage}
\begin{equation}
\label{eq:advtg}
    \maxadvtg \coloneqq \max_{G} \advtg(G)
\end{equation}
over all graphs with a fixed number of vertices.
This quantity measures the maximum speedup teleportation 
can provide on worst-case permutations for \emph{any} graph.
We also show tighter bounds on the advantage for some common classes of graphs.

We immediately have an upper bound on $\advtg$ from \cref{\swap{}n}. Since any 
teleportation algorithm must have depth $\bigomega{1}$, and a \swap{} algorithm
can implement any permutation in depth $\bigO{N}$, we have 
\begin{equation}
    \maxadvtg = \bigO{N}.
\end{equation}

We now show a tighter bound.

%-------------------------------------------------------------------------%
\subsection{Advantage for general graphs}
Combining \cref{cut_bound} and \cref{diam_bound}, we have
\begin{equation}
\rts{G} \geq \max\{\frac{2}{c(G)}-1, \diam(G)\}
\end{equation}
We now consider the relationship between $\diam(G)$ and $\frac{1}{c(G)}$. Intuitively, increasing
the diameter while keeping $N$ constant `stretches' the graph, tightening bottlenecks.
This causes $c(G)$ to decrease.
Similarly, eliminating bottlenecks in the graph requires adding more edges across cuts, 
thereby increasing the connectivity of the graph and reducing the diameter. 
We thus expect that graphs with higher diameter will have higher $\frac{1}{c(G)}$,
and graphs with small $\frac{1}{c(G)}$ will have small diameter. 
We can express this relation more precisely as follows.

\begin{restatable}{lemma}{diamBound}\label{lem:diamBound}
For any connected simple graph $G$,
\begin{equation}\label{eq:diamBound}
\diam(G) \leq 2\frac{\log\frac{N}{2}}{\log \left(1 + c(G)\right)} + 2.
\end{equation}
\end{restatable}

\begin{proof}
See \cref{diamcproof}.
\end{proof}

One might expect graphs with
large diameter to allow large speedups, since the diameter lower bound only applies
to \swap{} routing. However, as illustrated by
\cref{lem:diamBound}, graphs with large diameter also have tight bottlenecks, 
and therefore, by \cref{cut_bound}, are not likely to permit large speedups.

We now show our main results bounding the advantage. Our main technical result
bounds the advantage in terms of the diameter of the graph. We note that this bound also applies to the separation between \swap{}s and teleportation routing for any permutation, and not just the worst-case separation.
\begin{lemma}\label{advdiam}
$\advtg(G) =\bigO{\sqrt{N} + \diam(G)}$.
\begin{proof}
We construct a \swap{}-based protocol that can simulate
a single round of teleportation in depth 
$\bigO{\sqrt{N} + \diam(G)}$, thereby upper bounding the teleportation advantage.

A single round of a teleportation protocol performs teleportation along a set of paths. These paths must intersect no more than
a constant number of times per vertex, since there are only a constant number of 
ancillas per vertex.

For all paths from the teleportation protocol of length at most $\sqrt{N}$, we \swap{} along the paths in parallel. Since each vertex only has a constant number of paths going through
it, a qubit can move through every vertex in constant depth.
Therefore, these \swap{}s can be performed in depth $\bigO{\sqrt{N}}$.

For an $N$-vertex graph, the number of paths of length
at least $l$ that intersect at most a constant number of times is in
$O(N/l)$.
Therefore, since each long path corresponds to a single token, after routing along all paths of length at most $\sqrt{N}$, we have
$\bigO{\sqrt{N}}$ tokens left to route. By \cref{th:sparse}, 
this can be done in depth $\bigO{\sqrt{N} + \diam(G)}$.

We can thus simulate each teleportation round in depth
$\bigO{\sqrt{N} + \diam(G)}$, which completes the proof.
\end{proof}
\end{lemma}

Combining our results, we now have a bound on the advantage for \emph{any} graph.

\begin{theorem}
\label{adv_bound}
$\maxadvtg = \bigO{\sqrt{N \log N}}$.
\end{theorem}
\begin{proof}
First, combining \cref{\swap{}n} and \cref{cut_bound}, we have
\begin{equation}
    \advtg(G) = \bigO{N \cdot c(G)}.
    \label{eq:n_cut}
\end{equation}
Combining this bound with the bound from \cref{advdiam}, we have
\begin{equation}
\label{eq:advmin}
    \advtg(G) \leq \min \left\{\bigO{N \cdot c(G)}, 
    \bigO{\sqrt{N} + \diam(G)} \right\}.
\end{equation}
We know that $c(G) > 0$.
Using the fact that $\log(x) \ge 1 - 1/x$ for $x > 0$, we have % 1/(1-1/(1+c)) = 1/(c/(1+c)) = 1/c + 1
\begin{equation}
    \frac{1}{\log(1+c(G))} \le \frac{1}{c(G)} + 1 = \bigO*{\frac{1}{c(G)}},
\end{equation}
where in the last equality we used $c(G) \leq 1$.
Applying this to \cref{lem:diamBound} and \cref{eq:advmin}, we have
\begin{equation}
    \label{eq:advmin2}
    \advtg(G) \leq\min \left\{ \bigO{N \cdot c(G)},
    \bigO*{\sqrt{N} + \frac{\log(N)}{c(G)}} \right\}.
\end{equation}
Recall the definition of the maximum teleportation advantage from \cref{eq:advtg}:
\begin{equation}
    \maxadvtg \coloneqq \max_{G} \advtg(G).
\end{equation}
Therefore,
\begin{equation}
\begin{aligned}
    \maxadvtg \leq \max_{G} \min \bigg\{ & \bigO{N \cdot c(G)}, \\
      & \bigO*{\sqrt{N} + \frac{\log(N)}{c(G)}} \bigg\}.
\end{aligned}
\end{equation}
As $c(G)$ varies, the two bounds in the minimum vary inversely. The first bound,
from \cref{eq:n_cut}, is monotonically increasing in $c(G)$ for $c(G) \in (0,1]$. 
The second bound is monotonically decreasing in $c(G)$ for $c(G) \in \left( 0,\frac{\log N}{\sqrt{N}} \right]$.
Note that when $c(G) \sim 1/N$ (recall that $c(G) \geq 2/N$), the first bound is smaller, while when $c(G) \sim \frac{\log N}{\sqrt{N}}$,
the second bound is smaller. The largest minimum of the two bounds is thus obtained when they are equal.

The minimum of the two bounds is thus maximized when $c(G) = \sqrt{(\log N) /N}$. Note that even if a graph with $c(G) = \sqrt{(\log N) /N}$ does not exist, any other value of $c(G)$ will result in a smaller right-hand side of \cref{eq:advmin2}. With $c(G) = \sqrt{(\log N) /N}$, we obtain 
\begin{equation} 
    \maxadvtg = \bigO*{\sqrt{N\log N}}
\end{equation}
as claimed.
\end{proof}

This bound applies to any graph, and is thus independent of the diameter of the graph. 
Therefore, this 
result shows that in graphs with diameter $\omega(\sqrt{N\log N})$, 
we cannot obtain a routing time separation between teleportation- and \swap{}-based routing that is proportional to the diameter.

%-------------------------------------------------------------------------%

Next we show tighter bounds for a few common families of graphs.

\subsection{Grids}
For $d$-dimensional grids (i.e, $P_{n}^{\square d}$, the $d$-fold Cartesian product of the path graph $P_n$, with $N = n^d$ vertices), 
the vertex cut bound (\cref{cut_bound}) gives
\begin{equation} \label{eq:grid}
    \rtl{P_{n}^{\square d}} \geq \frac{2}{c(P_{n}^{\square d})} - 1 \geq n -1,
\end{equation}
where $c(P_{n}^{\square d}) \leq 2/n$ follows from considering a hyperplane that bisects the grid along one dimension.
From~\cite{Alon1994}, we have
\begin{equation}
    \rts{G_1 \square G_2} = 2\rts{G_1} + \rts{G_2}.
\end{equation}
Therefore, the \swap{} routing time of a $d$-dimensional grid is $\bigO{dN^{1/d}} = \bigO{d n}$. For constant $d$,
this saturates the cut bound in \cref{eq:grid}. Therefore, there is no worst-case speedup from either teleportation or full LOCC,
i.e, $\advtg(P_{n}^{\square d}) = 1$.

%-------------------------------------------------------------------------%
\subsection{Expander graphs}
We bound the advantage for spectral expander graphs to be $\poly(\log N)$.
The \emph{(normalized) Laplacian} of a graph, $G$, is defined as
\begin{equation}
    \lapl_{u,v} =
    \begin{cases}
        1 &\text{if } u = v \\
        -\frac{1}{\sqrt{d_v d_u}} &\text{if } (u,v) \in E(G) \\
        0 &\text{otherwise},\\
    \end{cases}
\end{equation}
where $d_v$ is the degree of vertex $v$.
The matrix $\lapl$ is symmetric and positive
such that we can order its eigenvalues as
$0 = \lambda_0 \le \lambda_1 \le \dots \le \lambda_{n-1}$.
We write $\lambda(G)$ for $\lambda_1$ of the Laplacian of $G$.
\emph{Spectral expander graphs} are graphs of bounded degree with $\lambda(G) = \bigomega{1}$.
For a comprehensive introduction to spectral graph theory, consult~\cite{Chung1996}.

To bound the advantage for spectral expander graphs, we
first use the following upper bound on the swap-based routing number.
Let $d_{*} \coloneqq 
\frac{\max_{v \in V} d_v}{\min_{v \in V} d_v}$ denote the degree ratio of a graph.

\begin{theorem}[\cite{quantum_routing}]
\label{expander}
For any graph $G$ and permutation $\pi$, 
\begin{equation}
    \rts{G, \pi} = \bigO*{\frac{d_{*}}{\lambda(G)^2}\log^2 N}.
\end{equation}
\end{theorem}

Combining this result with the lower bound of \cref{cut_bound}, we immediately get
\begin{equation}
    \advtg(G) =  \bigO*{\frac{d_{*} c(G) \log^2 N}{\lambda(G)^2}}.
\end{equation}
Thus graphs with $\lambda(G) = \bigomega{1}$ and $d_* = \bigO{1}$ (such as spectral expanders)
have at most a polylogarithmic advantage.

%-------------------------------------------------------------------------%
\subsection{Hypercubes}
The \swap{}-based routing time for a $d$-dimensional
hypercube $Q_d$ is~\cite{Alon1994,hypercubes}
\begin{equation}
    \rts{Q_d} = \Theta(d).
\end{equation}
Since $|V(Q_d)| = N = 2^d$,
\begin{equation}
    \rts{Q_d} = \log N.
\end{equation}
Now, we will show that $c(Q_d) = \bigtheta{\frac{1}{\sqrt{d}}}$.
In a hypercube, Hamming balls (i.e., sets of all points with Hamming weight $\leq r$ for some integer $r$) have the smallest boundary of all sets of a given size~\cite{harper}.
Taking the Hamming ball of radius $d/2$ as $X$, we have $|X| = 2^{d-1}$
and $|\delta X| = \binom{d}{d/2} = \bigtheta{2^d/\sqrt{d}}$.
Therefore, $c(G) =\bigtheta{1/\sqrt{d}}$.
Using \cref{cut_bound}, we have
$\rtt{G} = \bigomega{\sqrt{d}} = \bigomega{\sqrt{\log N}}$.
Teleportation thus offers
at most an $\bigO{\sqrt{\log N}}$ advantage on hypercubes.

%-------------------------------------------------------------------------%

\subsection{Other graphs}

The cyclic butterfly graph $B_r$ has been proposed as a 
constant-degree interaction graph that allows for fast circuit synthesis~\cite{cowtan, butterfly}.
Each of the $N = r2^r$ vertices is labelled $(w,i) \in \{0,1\}^r \times [r]$.
Vertices $(w,i)$ and $(v, i+1 \bmod r)$ are connected if $w=v$ 
or if $w$ and $v$ differ by exactly one bit in the $i$th position.
The cyclic butterfly has diameter $\bigO{\log N}$, degree 4, and 
$\rtsop(B_r) = \bigO{\log N}$~\cite{butterfly}.

We now show that the $\bigO{\log N}$ protocol is optimal even for teleportation routing on the cyclic butterfly graph, so 
$\advtg(G) = \bigO{1}$. 
Bipartition the vertices into sets $X, \overline{X}$ such that $X$ consists of all rows with bit $j = 0$ for some $j$, and $\overline{X}$ consists of all rows with bit $j = 1$. For this partition, $|X| = r2^{r-1}$ and $|\delta X| = 2^r$, so $c(G) \leq 2/r$.
Since $r = \Theta(\log N)$, $c(G) = \bigO{\frac{1}{\log N}}$, so from 
\cref{cut_bound}, $\advtg(G) = \bigO{1}$. 

The complete graph $K_N$ has $\rtsop(K_N) = \bigO{1}$, and therefore has $\advtg(K_N)=1$.

Finally, graphs with poor expansion properties---in particular, with vertex expansion
$c(G) = \bigO{\frac{\poly(\log N)}{N}}$---have at most polylogarithmic advantage by \cref{eq:n_cut}.
%-------------------------------------------------------------------------%
\section{Discussion}
\label{discussion}
In this paper, we have used quantum teleportation to speed up the task of
permuting qubits on graphs. We have shown examples of specific types of permutations
that can be sped up by teleportation. Further, we have shown an example of a graph
that exhibits a worst-case teleportation routing speedup of $\log N$. 
Our main technical result (\cref{adv_bound}) is a general upper bound of $\bigO{\sqrt{N\log N}}$ 
on the worst-case routing speedup.
We also show that many practical architectures cannot implement arbitrary interactions with low overhead, even with fast LOCC (unlike previous work which only considered unitary evolution). Such a negative result provides useful constraints for the design of quantum devices, suggesting that designing new architectures may prove fruitful.

Our work leaves an open question on whether there exists a graph with $\advtg(G) = \omega(\log N$).
Such a graph cannot be a spectral expander graph as per \cref{expander}.
From \cref{lem:diamBound}, we know that a graph with large diameter will have poor expansion properties (small $c(G)$) and therefore will not have a large teleportation advantage as per \cref{eq:n_cut}. Some candidate graphs for a superlogarithmic teleportation advantage are those with $c(G) \approx \sqrt{(\log N) /N}$. Such graphs may come closer to achieving a teleportation advantage given by the upper bound of \cref{adv_bound}.

Furthermore, we believe that there should exist a tighter upper bound than~\cref{adv_bound} on the maximum teleportation advantage for any graph.
This is one particularly interesting direction in resolving the advantage of a teleportation protocol over \swap{}s.
There could be more sophisticated methods that give tighter bounds by exploiting parallelism. 
A possible approach to tightening this bound would be to show a \swap{} protocol
that performs routing from multiple teleportation
rounds in parallel, since \swap{} paths need not obey the strict conditions of
teleportation paths (namely, allowing only a constant number of path
intersections per vertex).

We have primarily focused on the teleportation model of routing. However,
teleportation routing is a special case of the more general LOCC model of routing. We currently
do not know whether the full power of LOCC can provide a super-constant speedup over
teleportation routing. This is analogous to another open question, namely whether routing
with arbitrary 2-qubit gates---or even with arbitrary bounded 2-qubit Hamiltonians---can provide a super-constant speedup over \swap{}-based routing~\cite{quantum_routing}. 

Herbert~\cite{herbert} posed the question of establishing to
what extent ancillas can be used to reduce the routing depth.
Rosenbaum~\cite{rosenbaum} showed an $\bigO{1}$ routing protocol
on $N$ qubits with $\bigO{N^2}$ ancillas (i.e., an advantage of $\bigO{N}$), while systems without ancillas cannot perform LOCC or teleportation routing, and therefore cannot exhibit any speedups. We have investigated an intermediate regime, and have shown that a linear number of ancillas cannot allow for speedups
greater than $\bigO{\sqrt{N \log N}}$. It remains an open question to further investigate the 
space-time tradeoff between the number of ancilla qubits and the routing time.

We assume noiseless circuits, but in the presence of noise the performance of teleportation protocols depends directly on the fidelity of the required resource Bell pairs. We are primarily interested in ways to use teleportation for routing, and Bell pairs are necessary for this process.  Our current teleportation routing model does not distinguish between routing over long or short paths, but a more comprehensive model of routing could prioritize shorter paths as they will be less error prone without error correction.  Alternatively, we could use a purification protocol \cite{bennett_purification} to prepare high-fidelity Bell pairs at the cost of additional ancillas and overhead, or we could encode our state in an error-correcting code \cite{kitaevTopo} to suppress the error rate when operating between nodes. If operating in a quantum network, we can make use of protocols generalizing entanglement swapping from Bell basis measurements to $n$-qubit GHZ states to improve the performance of repeater protocols in lossy quantum networks \cite{Patil_entanglement_generation}. Alternatively, we can prepare a high-fidelity Bell pair by performing multiple repeater protocols along different paths in parallel \cite{Pant_2019} or using multiplexers on each edge \cite{Lee2022}. 

A more general task than routing is to perform
unitary synthesis, i.e, decompose a particular unitary into 2-qubit gates that
can be applied on our locality-constrained qubits. It remains an open question
to understand how much unitary synthesis can be sped up by using LOCC with
a linear number of ancillary qubits. Previous work has shown an $\bigomega{N}$
speedup for implementing fanout~\cite{Pham2013} and preparing GHZ and W states~\cite{Piroli}, and an $\bigomega{\sqrt{N}}$ speedup for preparing toric code states~\cite{Piroli}, which takes time $\bigomega{\sqrt{N}}$ without LOCC~\cite{bravyi}. Previous work has also shown how measurements of cluster states can be used to efficiently prepare long-range entanglement~\cite{spt} and states with exotic topological order~\cite{verresen}.
In principle, LOCC could provide superlinear speedups for unitary synthesis, as we
currently have no upper bounds on the advantage for arbitrary unitaries.

%-------------------------------------------------------------------------%
%-------------------------------------------------------------------------%
\section*{Acknowledgements}
We thank Andrew Guo, Yaroslav Kharkov, Samuel King, and Hrishee Shastri for helpful discussions.
D.D.\ acknowledges support by the NSF GRFP under Grant No.~DGE-1840340, an LPS Quantum Graduate Fellowship, and the U.S. Department of Energy, Office of Science, Office of Advanced Scientific Computing Research, Quantum Testbed Pathfinder program (award number DE-SC0019040).
A.B.~and A.V.G.~were supported in part by ARO MURI,  DoE ASCR Quantum Testbed Pathfinder program (awards No.~DE-SC0019040 and No.~DE-SC0024220), NSF QLCI (award No.~OMA-2120757), DoE ASCR Accelerated Research in Quantum Computing program (award No.~DE-SC0020312), NSF STAQ program, DARPA SAVaNT ADVENT, AFOSR, AFOSR MURI, and U.S.~Department of Energy Award No.~DE-SC0019449. Support is also acknowledged from the U.S.~Department of Energy, Office of Science, National Quantum Information Science Research Centers, Quantum Systems Accelerator.  
A.M.C.\ and E.S.\ acknowledge support by the U.S.\ Department of Energy,
Office of Science, Office of Advanced Scientific Computing Research,
Quantum Testbed Pathfinder program (award number DE-SC0019040) 
and the U.S.\ Army Research Office (MURI award number W911NF-16-1-0349).
E.S.\ acknowledges support from an IBM PhD Fellowship
and the U.S. DoE, Office of Science NQISRC, Quantum Science Center for finalizing and publishing the paper.
%-------------------------------------------------------------------------%
%-------------------------------------------------------------------------%
\bibliography{sources}

\appendix

\section{Sparse routing}
\label{sparse_tree}
Previous work~\cite{quantum_routing} shows an $\bigO{\diam(G)+k^2}$ \swap-based routing algorithm to route $k$ vertices
on a graph $G$. In this Appendix, we show that using \swap s with a constant number of ancillas per qubit, this result can be improved to be linear in $k$.

\begin{figure}
\subfloat[Line 2]{\includegraphics[width=0.3\columnwidth]{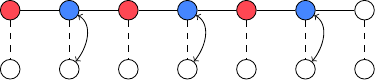}}
\quad
\subfloat[Line 4 ]{\includegraphics[width=0.3\columnwidth]{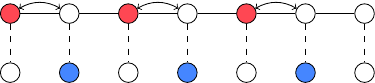}}
\quad
\subfloat[Line 5]{\includegraphics[width=0.3\columnwidth]{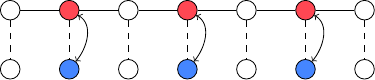}}
\hfill %%
\subfloat[Line 7]{\includegraphics[width=0.3\columnwidth]{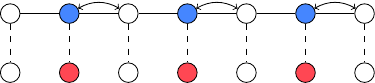}}
\quad
\subfloat[Line 8]{\includegraphics[width=0.3\columnwidth]{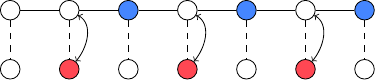}}
\quad
\subfloat[Advanced train]{\includegraphics[width=0.3\columnwidth]{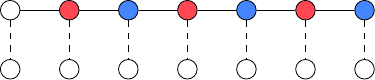}}
\caption{Advancing a train of tokens, as in \cref{alg:train}. Blank vertices hold state $\ket{0}$. The dashed lines represent connections with the local ancilla qubits.}
\label{fig:trainline}
\end{figure}

\begin{figure}
\subfloat[The blue train advances one vertex at a time.]{\includegraphics[width=0.45\columnwidth]{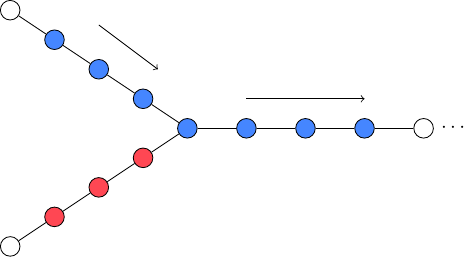}}
\subfloat[The red train waits for the blue train, but remains in the same connected token cluster.]{\includegraphics[width=0.45\columnwidth]{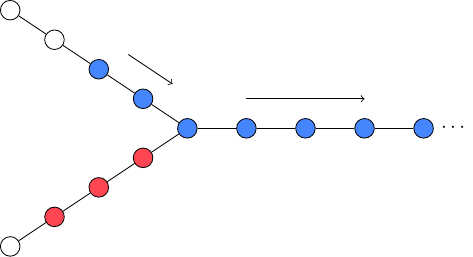}}
\hfill %%
\subfloat[The red train is concatenated to the blue train.]{\includegraphics[width=0.45\columnwidth]{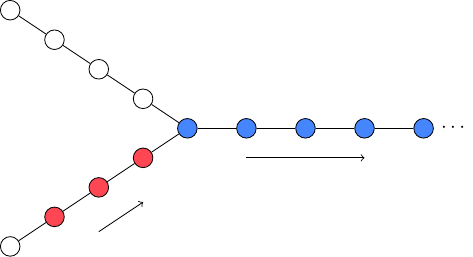}}
\subfloat[Once concatenated, they move as a single train.]{\includegraphics[width=0.45\columnwidth]{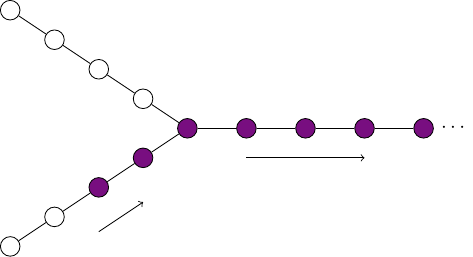}}
\caption{Joining trains of tokens.}
\label{fig:trains}
\end{figure}

We first introduce the following definitions.

\begin{definition}[Null token]
A \emph{null token} is a dummy token that can be routed anywhere.
In quantum routing, 
all ancillas are initialized with a null token in state $\ket 0$.
\end{definition}

\begin{definition}[Train]
A \emph{train} is a set of non-null tokens along a path subgraph of $G$.
\end{definition}

\begin{algorithm}[tbp]
\caption{Advance a train}\label{alg:train}
\Input{%
    Train $T$
        from vertices $0$ to $l-1$ on a path. Vertex $l$ has a null token. The data token on vertex $i$ is data$(i)$, and the token on the corresponding ancilla is ancilla$(i)$.
}
\ParFor{$i = 1, 3, \dots$}{%
    \swap{} data$(i)$ with ancilla$(i)$\;
}
\ParFor{$i = 0,2, 4, \dots$}{%
    \swap{} data($i$) with data($i+1$)\;
    \swap{} data($i+1$) with ancilla$(i+1)$\;
}
\ParFor{$i = 1,3,\dots$}{%
    \swap{} data($i$) with data($i+1$)\;
    \swap{} data($i$) with ancilla($i$)\;
}
\end{algorithm}

Now we show how a train can \emph{advance}, i.e., translate by 1 along its length.

\begin{lemma}\label{trains}
A train can advance in depth 5.
\end{lemma}

\begin{proof}
Suppose we want to move a train of length $l$ towards some vertex $r$.
We define the \emph{head} of a train as the token on the vertex closest to $r$, and the \emph{tail} as the token on the vertex furthest from $r$.
Consider the path subgraph spanned by the vertices the train lies on as well as the vertices of the shortest path from the head to $r$.
Let the tail lie on vertex $0$, and head lie at vertex $l-1$.
We use \cref{alg:train} to advance the train such that after 5 time steps, the tail of the train is at vertex $1$ and the head at $l$. This procedure is depicted in \cref{fig:trainline}.
\end{proof}

\begin{algorithm}[tbp]
\caption{Token cluster movement}\label{alg:tokencluster}
\Input{%
    Set of token clusters; vertex $r$.
}
In each token cluster, advance the train with head closest to $r$ by 1 using \cref{alg:train}.\;
If any two token clusters are adjacent, join them as a single token cluster.\;
If the head of any train $T_1$ is adjacent to the tail of another train $T_2$, join them as a single train with the head of $T_2$ and the tail of $T_1$.\;
\end{algorithm}

We now define a token cluster.

\begin{definition}[Token cluster]
A \emph{token cluster} is a set of trains such that each train contains a token on a vertex that is adjacent to a vertex with a token from another train in the token cluster.
\end{definition}
Token clusters move as in \cref{fig:trains}, by \cref{alg:tokencluster}.
Once a train joins a token cluster, it remains connected and part of the token cluster. 

We now prove \cref{th:sparse}, which we reproduce here for clarity.

\srt*

\begin{proof}
Let us call the $k$ tokens on vertices
\begin{equation}
    \{v \in V(G) \mid \pi(v) \ne v\}
\end{equation}
\emph{marked tokens},
and let the remaining token be \emph{unmarked tokens}.
Our algorithm involves three phases.

\textbf{Phase 1:}  
First, we \swap{} the unmarked tokens
into local ancilla qubits and store them there for the duration of routing. 
Every vertex that initially held an unmarked token now holds a null token.

Next, we select some vertex $r$ of $G$ arbitrarily (in practice, selecting $r$ to be at the center of the graph may provide constant-factor speedups). We now move all marked tokens towards vertex $r$ by \swap ping
along the shortest possible paths, until the tokens span a set of vertices
forming a tree connected to $r$. 
The tokens are moved in parallel, and when their paths intersect, the tokens move as trains, as per \cref{trains}.
When the paths of multiple trains intersect, they form a token cluster, and can be moved as in \cref{alg:tokencluster}.

Any given train is at most $\diam(G)$ distance away from $r$ at the start of Phase 1.
At every time step, a train either advances by 1 vertex towards $r$, or is part of a token cluster in which another train closer to $r$ advances.
Therefore, every token cluster becomes connected to $r$ in depth
$\bigO{\diam(G)}$, since in every token cluster, at least 1 train must reach $r$ in depth $\bigO{\diam(G)}$.
In particular, in $\bigO{\diam(G)}$ depth, all non-null tokens must span a tree containing $r$,
and thus have merged into a single token cluster.

\textbf{Phase 2:} Now we have $k$ vertices spanning a tree $T$.
Suppose token $v$ is mapped to the vertex $t(v)$ in $T$ after Phase 1.
Note that the token $u$ that was originally at $t(v)$ must also be a
marked token, and therefore must now lie in $T$.
We route the tokens on $T$ 
according to a permutation $\pi'$ such that 
\begin{equation}
    \pi'(t(v)) \coloneqq t(\pi(v))
\end{equation}
for all $t(v) \in V(T)$,
in depth $2k$~\cite{Alon1994}.

\textbf{Phase 3:} We now simply perform Phase 1 in reverse.
During Phase 1, the marked token at $u$ was mapped to $t(u)$. Therefore, after Phase 3,
the token at $t(u)$ is mapped to vertex $u$.
Therefore, the following mapping is applied to all vertices with marked tokens:
\begin{equation}
    u \xrightarrow{\text{Phase 1}} t(u) \xrightarrow{\text{Phase 2}} t(\pi(u))
    \xrightarrow{\text{Phase 3}} \pi(u).
\end{equation}

The combined depth of the three phases is at most
$\bigO{k+\diam(G)}$.
\end{proof}

%-------------------------------------------------------------------------%
%-------------------------------------------------------------------------%
\section{Proof of diameter-expansion trade-off}\label{diamcproof}

In this appendix, we prove \cref{lem:diamBound}, adapting Proposition 3.1.5 from~\cite{kowalski_expander} to vertex neighborhoods rather than edge neighborhoods.

\diamBound*
\begin{proof}
For any vertex $v\in V$, denote by $C(v,k)$ the set of all vertices that are at distance $k$ from $v$. We call $C(v,k)$ a circle of radius $k$ centered on $v$. Note that $C(v,k)\cap C(v,k') = \emptyset$ when $k\neq k'$. Next, define
\begin{equation}
    D(v,k) \coloneqq \bigcup\limits_{r=0}^{k} C(v,r)
\end{equation}
to be the disk of radius $k$ centered on $v$. Observe that $C(v,0) = D(v,0) = \{v\}$. Finally, choose an integer $\rho(v)$ such that $|D(v,\rho(v))|\le N/2 < |D(v,\rho(v)+1)|$ and call it the \emph{horizon} of $v$. For any vertex, a horizon exists and is an integer between 0 and $\diam(G)-1$.

By definition, for all $k\le \rho(v) + 1$, we have
\begin{equation}
    |C(v,k)|\ge c(G)\cdot |D(v,k-1)|.
\end{equation}
Applying this inequality gives
\begin{align}
    |D(v,\rho(v))| &= |C(v,\rho(v))| + |D(v, \rho(v)-1)|\\
        &\ge (1+c(G))|D(v, \rho(v)-1)|.
\end{align}
Recursing until we reach the base case $D(v, 0) = \{v\}$,
we obtain
\begin{equation}
    N/2 \ge (1+c(G))^{\rho(v)},
\end{equation}
giving
\begin{equation}\label{eq:horizon}
    \rho(v) \le \frac{\log(N/2)}{\log(1+c(G))}.
\end{equation}
Next, for any two vertices $u,v\in V$, let $d(u,v)$ denote the distance between $u,v$. We claim that
\begin{equation}
    d(u,v) \leq \rho(u) + \rho(v) + 2 .
\end{equation}
To see this, note that by definition, $|D(u,\rho(u)+1)|> N/2$ and $|D(v,\rho(v)+1)|> N/2$, which implies that $D(u,\rho(u)+1)\bigcap D(v,\rho(v)+1) \neq \emptyset$ by the pigeonhole principle. Therefore, there exists a vertex $t$ such that $d(u,t)\leq \rho(u)+1$ and $d(t,v)\leq \rho(v)+1$. By the triangle inequality, we have $d(u,v) \leq \rho(u) + \rho(v) + 2$ as claimed. 

Finally, we use \cref{eq:horizon} and maximize the distance over all vertex pairs $u,v$ to get
\begin{equation}
    \label{eq:final_diam_c}
    \diam(G) \le \frac{2\log(N/2)}{\log(1+c(G))} + 2
\end{equation}
as claimed.
\end{proof}

\end{document}